\documentclass[12pt]{article}
\usepackage[UKenglish]{babel} % setting language
\usepackage[T1]{fontenc} % for correct display of accented characters
\usepackage[utf8]{inputenc} % for correct input of accented characters
\usepackage{amsmath}
\usepackage{amsthm}
\usepackage{amssymb}
\usepackage{cite} % configuring citation
\usepackage{bm} % bold 3D vectors
\usepackage{indentfirst} % indent first paragraphs of sections
\usepackage{caption} % configuring captions for images
\usepackage{subcaption} % captioning subfigures
\usepackage{enumitem} % to specify enumerate styles
\usepackage{comment} % allows comment sections
\usepackage{todonotes} % allows TODO notes 
\usepackage{csquotes}
\usepackage{hyperref} % hyperref, load at the end
\newtheorem{theorem}{Theorem} % Theorem environment

\newtheorem{proposition}{Proposition}

\newtheorem*{assumption*}{Assumption}
\newtheorem{theoremA}{Theorem}

\newtheorem*{claim*}{Claim}
\makeatletter
 \DeclareFontFamily{OMX}{MnSymbolE}{}
 \DeclareSymbolFont{MnLargeSymbols}{OMX}{MnSymbolE}{m}{n}
 \SetSymbolFont{MnLargeSymbols}{bold}{OMX}{MnSymbolE}{b}{n}
 \DeclareFontShape{OMX}{MnSymbolE}{m}{n}{
	 <-6>  MnSymbolE5
	<6-7>  MnSymbolE6
	<7-8>  MnSymbolE7
	<8-9>  MnSymbolE8
	<9-10> MnSymbolE9
   <10-12> MnSymbolE10
   <12->   MnSymbolE12
 }{}
 \DeclareFontShape{OMX}{MnSymbolE}{b}{n}{
	 <-6>  MnSymbolE-Bold5
	<6-7>  MnSymbolE-Bold6
	<7-8>  MnSymbolE-Bold7
	<8-9>  MnSymbolE-Bold8
	<9-10> MnSymbolE-Bold9
   <10-12> MnSymbolE-Bold10
   <12->   MnSymbolE-Bold12
 }{}
 
 \let\llangle\@undefined
 \let\rrangle\@undefined
 \DeclareMathDelimiter{\llangle}{\mathopen}%
					  {MnLargeSymbols}{'164}{MnLargeSymbols}{'164}
 \DeclareMathDelimiter{\rrangle}{\mathclose}%
					  {MnLargeSymbols}{'171}{MnLargeSymbols}{'171}
\makeatother

\usepackage{tikz}
\usetikzlibrary{decorations.pathmorphing}
\usetikzlibrary{calc}

\usepackage{pgfplots}
\pgfplotsset{width=10cm,compat=1.9}
\usepgfplotslibrary{fillbetween}
\newcommand{\const}{\,{\rm const}\,}

\newcommand{\tr}{\,{\rm tr}\,}
\newcommand{\td}{\text{d}}

\newcommand{\ord}{\mathcal{O}}

\title{\textbf{Uniqueness of extremal charged  black holes in de Sitter}}
\author{David Katona\footnote{d.katona@sms.ed.ac.uk}
\\ \\ 
\small \sl School of Mathematics and Maxwell Institute for Mathematical Sciences, 
\\ 
\small \sl University of Edinburgh, King's Buildings, Edinburgh, EH9 3FD, UK }

\date{}

\begin{document}
%\begin{titlepage}
	\maketitle
	\begin{abstract}
		We prove a uniqueness theorem for the charged Nariai black holes and ultracold black holes in four dimensions. In particular, we show that an analytic solution to four-dimensional Einstein--Maxwell theory with a positive cosmological constant containing a static extremal Killing horizon with spherical cross-sections of large radius (compared to the cosmological scale), must be locally isometric to the extremal Reissner--Nordstr\"om--de Sitter black hole or its near-horizon geometry. The theorem generalises to extremal static horizons with small radius, establishing uniqueness of cold black holes for generic values of the radius.
	\end{abstract}	
%\end{titlepage}

\section{Introduction}

Black hole solutions of Einstein--Maxwell theory obey the famous no-hair theorem, which establishes uniqueness of the Kerr--Newman family among asymptotically flat, analytic black hole solutions with a connected, non-degenerate horizon (for a review see e.g.~\cite{chrusciel_stationary_2012}). The theorem has been extended to extremal Kerr and Kerr--Newman black holes~\cite{amsel_uniqueness_2010, figueras_uniqueness_2009, kleinwachter_analytical_2008, chrusciel_uniqueness_2010}, however if one allows for multi-black hole solutions, a similar uniqueness no longer holds. If one assumes staticity or supersymmetry, the general extremal black hole solution belongs to the Majumdar--Papapetrou family containing an arbitrary number of horizon components~\cite{chrusciel_nonexistence_2005,chrusciel_classification_2005,chrusciel_israel-wilson-perjes_2006}.

Much less is known if one considers a non-zero cosmological constant $\Lambda$. For the vacuum theory, uniqueness of the non-extremal Schwarzschild--de Sitter black hole among static solutions has been established under certain assumptions on the level sets of the lapse function~\cite{borghini_uniqueness_2023}, however recent numerical evidence shows that such assumptions are evaded by a class of binary black holes in de Sitter (dS)~\cite{dias_static_2023}. The existence of such binaries is in stark contrast with the $\Lambda=0$ case, when static non-extremal black holes with multiple horizons are excluded~\cite{bunting_nonexistence_1987}. In anti--de Sitter (AdS), a uniqueness theorem for non-extremal hyperbolic Schwarzschild--AdS spacetime with non-positive mass parameter has been established~\cite{chrusciel_towards_2001,lee_penrose_2015}, however for spherical black holes no analogous result is known\footnote{For some progress in this direction see \cite{chrusciel_nonsingular_2005,chrusciel_nondegeneracy_2017} and references therein.}. In the Einstein--Maxwell--$\Lambda$ theory, static non-extremal electro-vacuum black holes without any spatial symmetry have been numerically constructed in~\cite{herdeiro_static_2016}, suggesting a richer moduli space of black holes. This highlights the fact that our understanding of black holes with a non-zero cosmological constant is incomplete even in the static case.

Recently, the first uniqueness theorem for an extremal vacuum black hole with a cosmological constant has been proven~\cite{katona_uniqueness_2023}. In detail, it has been shown that the unique analytic solution to the four-dimensional\footnote{The theorem also applies to higher dimensional horizons with a maximally symmetric cross-sections.}  vacuum Einstein equations with a positive cosmological constant containing a static Killing horizon is the extremal Schwarzschild--de Sitter spacetime or its near-horizon geometry, the Nariai spacetime dS$_2\times$ S$^2$. Interestingly, this result excludes static extremal multi-black holes of this theory, at least in the analytic class. The purpose of the present work is to generalise this result for the electro-vacuum case with a positive cosmological constant.

Extremal black holes admit a well-defined near-horizon limit, in which the Einstein--equations for the intrinsic and extrinsic data of the horizon geometry decouple. The intrinsic data is described by the near-horizon geometry, which itself is a solution to the Einstein equations. Classification of such geometries has been achieved for a number of cases (for a review see \cite{kunduri_classification_2013}). However, in the near-horizon limit the extrinsic data (i.e. transverse derivatives of the metric) is lost, hence it is possible for a near-horizon geometry to correspond to multiple different black holes (or none). To recover this data, the systematic study of transverse deformations has been initiated~\cite{li_transverse_2016,li_electrovacuum_2019}, which aims to determine higher order corrections to the near-horizon geometry in the parameter describing the transverse direction to the horizon. First order transverse deformations have been studied for a number of horizons in the literature \cite{li_transverse_2016,li_electrovacuum_2019,kolanowski_towards_2021,dunajski_einstein_2016, fontanella_moduli_2016}. 

Remarkably, for simple enough geometries, it is possible to determine all higher order deformations. This method was used for establishing a uniqueness theorem for the extremal Schwarzschild--de Sitter black hole in vacuum gravity with a positive cosmological constant~\cite{katona_uniqueness_2023}. In the present work we aim to study all higher order transverse deformations to the extremal Reissner--Nordstr\"om--dS (RN--dS) spacetime in Einstein--Maxwell theory with a positive cosmological constant. This class of solutions contains three qualitatively different spacetimes and horizons, depending on the area $A_H$ of the spatial cross-sections of the horizon relative to the cosmological scale. If $A_H\Lambda>2\pi$, the near-horizon geometry is dS$_2\times$S$^2$, and the spacetime describes a black hole, known as the `charged Nariai black hole', with a qualitatively similar domain of outer communication (DOC) to the uncharged case. If $A_H\Lambda=2\pi$ the spacetime contains a degenerate horizon with flat near-horizon geometry $\mathbb{R}^{1,1}\times$S$^2$ (sometimes called the `ultracold black hole'). For $A_H\Lambda<2\pi$, the near-horizon geometry is AdS$_2\times$S$^2$, the spacetime is known as the `cold black hole'.

Our main result is summarised in the following theorem (for the detailed statement see Theorem \ref{thm_uniqueness}).

\begin{theoremA}
	Let $(M,g, F)$ be an analytic solution to the $d=4$ Einstein--Maxwell equations with a positive cosmological constant $\Lambda>0$. Let us further assume that the spacetime contains a degenerate Killing horizon with round spherical spatial cross-sections of area $A_H$, and that the Maxwell field is preserved by the Killing field generating the horizon. \begin{enumerate}[label=(\roman*)]
		\item If $A_H\Lambda\ge2\pi$, the solution is (up to identifications) the extremal Reissner--Nordstr\"om--dS solution or its near-horizon geometry dS$_2\times$S$^2$, or $\mathbb{R}^{1,1}\times$S$^2$ if the inequality is saturated,
		\item if $A_H\Lambda<2\pi$, for generic values of $A_H$ the solution is (up to identifications) the extremal Reissner--Nordstr\"om--dS solution or its near-horizon geometry AdS$_2\times$S$^2$.
	\end{enumerate}\label{thm_sum}
\end{theoremA}

The static near-horizon geometries of this theory have been completely determined \cite{kunduri_uniqueness_2008}. Assuming compactness of spatial cross-sections, all such near-horizon geometries have round spherical spatial geometry. Combining this fact with Theorem \ref{thm_sum} provides a uniqueness theorem for extremal RN--dS black holes (for `large' or generic `small' horizon area) among spacetimes containing a static degenerate Killing horizon with compact cross-sections. As in the vacuum case, this excludes static multi-black holes  among analytic solutions. This is in contrast to the $\Lambda=0$ case, when the Majumdar--Papapetrou solutions with multiple black holes are analytic. It is worth mentioning however, that it has been argued that generic extremal black holes with a cosmological constant should have lower regularity at the horizon \cite{horowitz_almost_2022}. 

Note that Theorem \ref{thm_sum} only relies on the near-horizon data without global assumptions on the spacetime (apart from analyticity). In the literature there are known examples when the horizon geometry alone determines the geometry of the whole spacetime. This includes four-dimensional extremal toroidal horizons in vacuum gravity \cite{moncrief_symmetries_1983}, three-dimensional vacuum solutions with a cosmological constant \cite{li_three_2014}, five-dimensional supersymmetric black holes in AdS \cite{lucietti_uniqueness_2021,lucietti_uniqueness_2022}, and most recently the extremal Schwarzschild--dS black hole \cite{katona_uniqueness_2023}.

The proof follows that of the uncharged black hole \cite{katona_uniqueness_2023}. It has been shown that the first order transverse deformation to the horizon of the extremal RN--dS black hole is unique \cite{kolanowski_towards_2021}. Here we establish that this uniqueness holds for higher order deformations if $A_H\Lambda\ge 2\pi$, and generic horizon areas if $A_H\Lambda<2\pi$. In more detail, we find that the Einstein--Maxwell equations at each order reduce to a system of coupled elliptic PDEs on the horizon for the spatial metric and a one-form, the trivial solution of which corresponds to the extremal RN--dS. By using a basis defined by spherical harmonics on the round sphere, we find that these equations admit non-trivial solutions only for a discrete set of horizon areas. As a result, uniqueness can be proven for all orders, except for a non-generic set of horizons. It is currently not known what solutions these additional deformations correspond to, if any.

The structure of the paper is as follows. In Section \ref{sec_RNdS} we review the extremal RN--dS solutions. In Section \ref{sec_NH} we introduce the notion of transverse deformations to a near-horizon geometry and determine the higher order deformations in the case of a static extremal horizon, proving our main result (Theorem \ref{thm_uniqueness}). In the Appendix we list the relevant components of the Einstein--Maxwell equations in Gaussian null coordinates.

\section{Charged extremal black holes in de Sitter} \label{sec_RNdS}

We will focus on black hole solutions of the four-dimensional Einstein--Maxwell equations with a positive cosmological constant $\Lambda$
\begin{align}
	R_{\mu\nu}&=\Lambda g_{\mu\nu} + 2\left(F_{\mu\kappa}F_\nu{}^\kappa-\frac{1}{4}g_{\mu\nu}F_{\kappa\sigma}F^{\kappa\sigma}\right)\;, \label{eq_EinsteinMaxwell}\\
	&\qquad\td \star \mathcal{F} =0\;,\qquad \td \mathcal{F}=0\;, \nonumber
\end{align}
where $\star$ denotes the Hodge star operator, $R_{\mu\nu}$ is the Ricci tensor, $g$ is the spacetime metric and $\mathcal{F}$ is the Maxwell field. The electrically charged\footnote{We only consider electrically charged RN-dS black holes, as magnetically charged and dyonic RN-dS black holes or their near-horizon data can be obtained by a global electromagnetic duality transformation.} Reissner--Nordstr\"om--dS (RN--dS) black holes are given by the metric and Maxwell field
\begin{align}
	g = -\Phi(r)\td t^2 + \frac{\td r^2}{\Phi(r)} &+ r^2\td\Omega_2^2\;, \qquad\qquad \mathcal{F} = -\frac{Q}{r^2}\td t\wedge \td r\;, \qquad\text{ with } \\
	&\Phi(r) = 1-\frac{m}{r}-\frac{\Lambda r^2}{3}+\frac{Q^2}{r^2}\;,\label{eq_RNdS}
\end{align}
where $\td\Omega_2^2$ denotes the unit round metric on S$^2$, $m>0$ is the mass parameter, and $Q$ is the electric charge of the black hole. For a generic black hole solution $\Phi(r)$ has three distinct positive zeros at $r_c, r_+, r_-$, corresponding to a cosmological, an outer, and an inner horizon, respectively. Two of these zeros coincide at $r=r_0$ if and only if
\begin{align}
	m = 2r_0\left(1-\frac{2}{3}\Lambda r_0^2\right)\;, \qquad Q^2 = r_0^2\left(1-\Lambda r_0^2\right)\;. \label{eq_eRNdS}
\end{align}
Thus, in this case the function $\Phi$ factorises as 
\begin{equation}
	\Phi(r)=-\frac{(r-r_0)^2(\Lambda r^2+2\Lambda r_0 r+3\Lambda r_0^2-3)}{3r^2}\;. 
\end{equation}

Note that the second equation of (\ref{eq_eRNdS}) restricts the dimensionless parameter $0\le \Lambda r_0^2\le 1$, and the upper bound is saturated for the extremal Schwarzschild--dS black hole. The moduli space of RN--dS black holes are depicted in Fig. \ref{fig_moduli}. The extremal RN--dS black holes are a one-parameter family of solutions parametrised by $r_0$, or equivalently the horizon area $A_H=4\pi r_0^2$. Depending on this parameter, the charged black holes belong to three qualitatively different classes of spacetimes, whose Penrose diagrams are depicted in Fig. \ref{fig_Penrose}. 

%Moduli space
\begin{figure}[h!]
	\centering
	\begin{tikzpicture}
		\begin{axis}[
			axis lines = left,
			xlabel = {\(m\Lambda^{1/2}\)},
			ylabel = {\(|Q|\Lambda^{1/2}\)},
			xtick distance = 0.2,
			ytick distance = 0.2,
			xmax = 1.4,
			ymax = 0.59,
		]
		\path[name path=axis] (axis cs:0,0) -- (axis cs:0.666667,0);
		\addplot[
			name path=f,
			domain=0:1,
			samples = 200,
			samples y=0,
		]
		({(2*sqrt(x)*(1-2*x/3))},
		{sqrt(x*(1-x))});
		\node [coordinate,pin={[align=center]above:{ultracold BH}}]
        at (axis cs:0.9428,0.5)   {};
		\node [coordinate,pin={[align=center]135:{cold BH}}]
        at (axis cs:0.50445,0.25514)   {};
		\node [coordinate,pin={[align=center]25:{charged Nariai BH}}]
        at (axis cs:0.714765,0.21894)   {};
		\node [coordinate,pin={[align=center]45:{extremal \\ Schwarzschild--dS BH}}]
        at (axis cs:0.66667,0)   {};
		\draw (axis cs:0.67165,0.0705337) -- (axis cs:0.666667,0);
		\node[circle, align=center] at (axis cs:0.35,0.06) {non-\\extremal BH};
		\addplot [
        thick,
        color=gray,
        fill=gray, 
        fill opacity=0.1
    ]
    fill between[
        of=f and axis,
        soft clip={domain=0:1},
    ];
		\end{axis}
		\end{tikzpicture}
\caption{Moduli space of RN--dS black holes.}
\label{fig_moduli}
\end{figure}
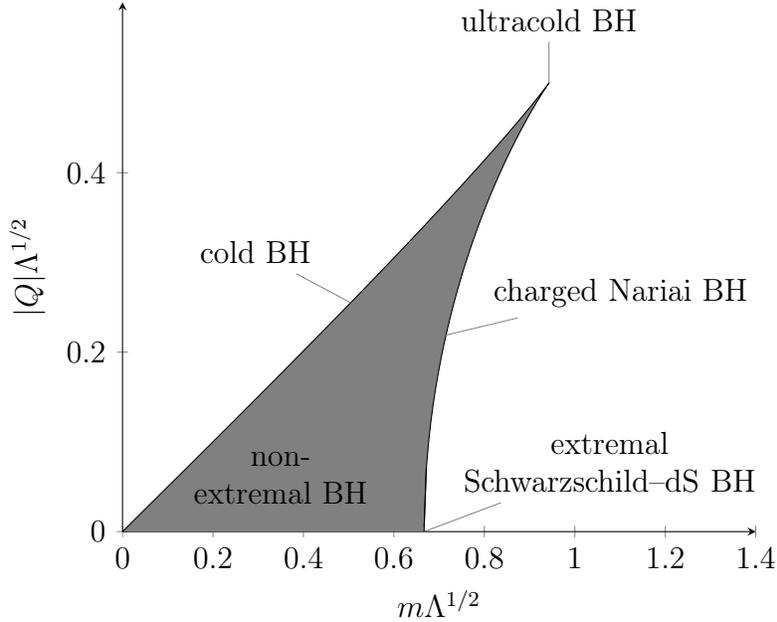 

For `large' horizon radius $1>r_0^2\Lambda>1/2$, the cosmological and outer horizons coincide $r_0=r_c=r_+$, and the extremal horizon has dS$_2\times$S$^2$ near-horizon geometry. The structure of the spacetime is depicted in Fig. \ref{subfig_Nariai} \cite{cardoso_nariai_2004}. Note that, similarly to the uncharged case \cite{podolsky_structure_1999}, one can check that points $P^\pm$ are asymptotic points where causal geodesics with $t=\const$ end (in the coordinates of (\ref{eq_RNdS})). Thus, the spacetime describes a black hole with respect to the future timelike infinity $P^+$, called the charged Nariai black hole. The region $r>r_0$ covers the DOC, which is qualitatively similar to that of the uncharged case, in particular the Killing vector field $\partial_t$ is spacelike for $r>r_-$ and $r$ acquires the role of time coordinate\footnote{The extremal RN--dS for $r_0^2\Lambda\ge1/2$ is not static in the strictest sense, since the hypersurface-orthogonal Killing vector field (which generates the horizon) is nowhere timelike in the DOC.}.

For `small' horizon radius $r_0^2\Lambda<1/2$, the inner and outer horizons coincide $r_0=r_+=r_-$. In this case the near-horizon geometry is AdS$_2\times$S$^2$. These spacetimes are known as cold black holes, and their Penrose diagrams can be seen in Fig. \ref{subfig_cold} \cite{cardoso_nariai_2004}. In contrast to the charged Nariai black hole, $\partial_t$ is timelike for $r<r_c$. 

%Penrose diagrams
\begin{figure}[h!]
	\centering
	\begin{subfigure}[c]{0.45\textwidth}
		\centering
\begin{tikzpicture}[scale=0.90]

	\coordinate (BHa) at (-4, 2);
	\coordinate (BHb) at (0, 2);
	\coordinate (BHc) at (4, 2);
	\coordinate (LeftTop) at (-4.5, 2);
	\coordinate (RightTop) at (4.5, 2);

	\path (BHa) +(-45:2.8284)  coordinate  (BHaBottom);
	\path (BHb) +(-45:2.8284)  coordinate  (BHbBottom);
	\path (BHaBottom) +(180:2.5)  coordinate  (LeftBottom);
	\path (BHbBottom) +(0:2.5)  coordinate  (RightBottom);
	\path (BHa) + (-135:0.75) coordinate (LeftHorizonEnd);
	\path (BHc) + (-45:0.75) coordinate (RightHorizonEnd);

	%singularities
	\draw[decorate,decoration=zigzag, thick] (LeftTop) --(BHa) -- node[midway, above] {\tiny$r=0$} (BHb) -- 
	node[midway, above] {\tiny$r=0$} (BHc)-- (RightTop);

	%horizons
	\draw (LeftHorizonEnd) -- (BHa) 
		-- node[midway, above, sloped] {\tiny$r= r_0$} (BHaBottom) 
		-- node[midway, above, sloped] {\tiny$r= r_0$} (BHb) 
		-- node[midway, above, sloped] {\tiny$r= r_0$} (BHbBottom) 
		-- node[midway, above, sloped] {\tiny$r= r_0$} (BHc)
		-- (RightHorizonEnd);

	%Scris
	\draw (LeftBottom) -- 
		node[near start, below] {\small$\cal{J}^-$}
		(BHaBottom) -- 
		node[midway, below] {\small$\cal{J}^-$} 
		node[midway, above] {\tiny$r = \infty$}
		(BHbBottom) -- 
		node[near end, below]{\small$\cal{J}^-$} 
		(RightBottom);

		\filldraw[color=black, fill=white, thick](BHa) circle (0.12);
		\filldraw[color=black, fill=white, thick](BHb) circle (0.12);
		\filldraw[color=black, fill=white, thick](BHc) circle (0.12);
		%\filldraw[color=black, fill=white, thick](BHaBottom) circle (0.12);
		%\filldraw[color=black, fill=white, thick](BHbBottom) circle (0.12);

		\path (BHa) +(90:0.35) node (Ia) {\tiny$P$};
		\path (BHb) +(90:0.35) node (Ia) {\tiny$P$};
		\path (BHc) +(90:0.35) node (Ia) {\tiny$P$};
		%\path (BHaBottom) +(-90:0.4) node (Ia) {\tiny$Q$};
		%\path (BHbBottom) +(-90:0.4) node (Ia) {\tiny$Q$};

\end{tikzpicture}
\caption{Extremal Schwarzschild--dS BH}
\label{subfig_SdS}
\end{subfigure}
\hfill
\begin{subfigure}[c]{0.45\textwidth}
		\centering
\begin{tikzpicture}[scale=0.90]

	%basic grid
	\coordinate (LB) at (-4,0);
	\coordinate (Q1) at (-2,0);
	\coordinate (Q2) at (2,0);
	\coordinate (RB) at (4,0);
	\coordinate (LSB) at (-4,0);
	\coordinate (P1) at (-4,2);
	\coordinate (P2) at (0,2);
	\coordinate (P3) at (4,2);
	\coordinate (bif1) at (-2,4);
	\coordinate (bif2) at (-4,4);
	\coordinate (P4) at (-4,6);
	\coordinate (P5) at (0,6);
	\coordinate (P6) at (4,6);
	\coordinate (LT) at (-4,8);
	\coordinate (Q3) at (-2,8);
	\coordinate (Q4) at (2,8);
	\coordinate (RT) at (4,8);

	%singularities
	\draw[decorate,decoration=zigzag, thick] (P1) --  
	node[pos=0, below] {\tiny$P^+$} 
	node[midway, above, sloped] {\tiny$r=0$}
	node[pos=1, above] {\tiny$P^-$} (P4);

	\draw[decorate,decoration=zigzag, thick] (P2) --  
	node[pos=0, below] {\tiny$P^+$} 
	node[midway, above, sloped] {\tiny$r=0$} 
	node[midway, below, sloped] {\tiny$r=0$}
	node[pos=1, above] {\tiny$P^-$} (P5);

	\draw[decorate,decoration=zigzag, thick] (P3) --  
	node[pos=0, below] {\tiny$P^+$} 
	node[midway, above, sloped] {\tiny$r=0$}
	node[pos=1, above] {\tiny$P^-$} (P6);
	%horizons
	\draw (Q1) -- (P1) -- 
	node[midway, above, sloped] {\tiny$r=r_-$} (Q4) -- (P6) -- 
	node[pos=0.84, above, sloped] {\tiny$r=r_0$}
	node[midway, below, sloped] {\tiny$r=r_-$} (Q1);
	\draw (Q2) -- (P3) -- 
	node[pos=0.75, above, sloped] {\tiny $r=r_-$} (P5)-- 
	node[midway, below, sloped] {\tiny $r=r_0$} (Q3) -- (P4) -- 
	node[midway, below, sloped] {\tiny$r=r_-$} 
	node[pos=0.84, above, sloped] {\tiny$r=r_0$} (Q2);

	%Scris
	\draw (LB) -- 
		node[near start, below] {\small$\cal{J}^-$}
		%node[pos=1, below] {\tiny$Q^-$}
		(Q1) -- 
		node[midway, below] {\small$\cal{J}^-$} 
		node[midway, above] {\tiny$r = \infty$}
		(Q2) -- 
		node[near end, below]{\small$\cal{J}^-$} 
		%node[pos=0, below] {\tiny$Q^-$}
		(RB);

	\draw (LT) -- 
		node[near start, above] {\small$\cal{J}^+$}
		%node[pos=1, above] {\tiny$Q^+$}
		(Q3) -- 
		node[midway, above] {\small$\cal{J}^+$} 
		node[midway, below] {\tiny$r = \infty$}
		(Q4) -- 
		node[near end, above]{\small$\cal{J}^+$} 
		%node[pos=0, above] {\tiny$Q^+$}
		(RT);
	%asymptotic points
	%\filldraw[color=black, fill=white, thick](Q1) circle (0.12);
	%\filldraw[color=black, fill=white, thick](Q2) circle (0.12);
	%\filldraw[color=black, fill=white, thick](Q3) circle (0.12);
	%\filldraw[color=black, fill=white, thick](Q4) circle (0.12);
	\filldraw[color=black, fill=white, thick](P1) circle (0.12);
	\filldraw[color=black, fill=white, thick](P2) circle (0.12);
	\filldraw[color=black, fill=white, thick](P3) circle (0.12);
	\filldraw[color=black, fill=white, thick](P4) circle (0.12);
	\filldraw[color=black, fill=white, thick](P5) circle (0.12);
	\filldraw[color=black, fill=white, thick](P6) circle (0.12);

\end{tikzpicture}
\caption{Charged Nariai BH}
\label{subfig_Nariai}
\end{subfigure}
\vskip\baselineskip
\begin{subfigure}[b]{0.45\textwidth}
	\centering
	\begin{tikzpicture}[scale=0.9]
		\coordinate (LB) at (-2,-2);
		\coordinate (RB) at (2,-2);
		\coordinate (LT) at (-2,2);
		\coordinate (RT) at (2,2);
		\coordinate (O) at (0,0);

		%singularities 
		\draw[decorate,decoration=zigzag, thick] (LB) --  node[midway, below, sloped]{\tiny $r=0$} (LT);
		\draw[decorate,decoration=zigzag, thick] (RB) --  node[midway, above, sloped]{\tiny $r=0$} (RT);
		%horizons
		\draw (LB) -- node[pos=0.25, below, sloped]{\tiny $r=r_0$} (RT);
		\draw (LT) -- node[pos=0.25, above, sloped]{\tiny $r=r_0$} node[pos=0.5, below]{\tiny $P$} (RB);
		%Scris
		\draw (LB)--node[midway, below]{\small$\mathcal{J}^-$}node[midway, above]{\tiny$r=\infty$} (RB);
		\draw (LT)--node[midway, above]{\small$\mathcal{J}^+$}node[midway, below]{\tiny$r=\infty$} (RT);
		%asymptotic points
		\filldraw[color=black, fill=white, thick](O) circle (0.12);
	\end{tikzpicture}
	\caption{Ultracold BH}
	\label{subfig_UC}
\end{subfigure}
\hfill
\begin{subfigure}[b]{0.45\textwidth}
	\centering
	\begin{tikzpicture}[scale=0.9]
		\coordinate (LB) at (-4,0);
		\coordinate (Q1) at (-2,0);
		\coordinate (Q2) at (2,0);
		\coordinate (RB) at (4,0);
		\coordinate (LM) at (-4,2);
		\coordinate (RM) at (4,2);
		\coordinate (LT) at (-4,4);
		\coordinate (Q3) at (-2,4);
		\coordinate (Q4) at (2,4);
		\coordinate (RT) at (4,4);
		%singularities 
		\draw[decorate,decoration=zigzag, thick] (LB) --  node[pos=0.25, below, sloped]{\tiny $r=0$}node[pos=0.75, below, sloped]{\tiny $r=0$} (LT);
		\draw[decorate,decoration=zigzag, thick] (RB) --  node[pos=0.25, above, sloped]{\tiny $r=0$}node[pos=0.75, above, sloped]{\tiny $r=0$} (RT);
		%horizons
		\draw (LM) -- 
		node[midway, below, sloped]{\tiny $r=r_0$} (Q3)--
		node[pos=0.25, below, sloped]{\tiny $r=r_c$} (Q2)--
		node[midway, above, sloped]{\tiny $r=r_0$}(RM)--
		node[midway, below, sloped]{\tiny $r=r_0$}(Q4)--
		node[pos=0.75, above, sloped]{\tiny $r=r_c$} (Q1)--
		node[midway, above, sloped]{\tiny $r=r_0$}(LM);
		%Scris
		\draw (Q1)--node[midway, below]{\small$\mathcal{J}^-$} node[midway, above]{\tiny$r=\infty$} (Q2);
		\draw (Q3)--node[midway, above]{\small$\mathcal{J}^+$} node[midway, below]{\tiny$r=\infty$}(Q4);
		%asymptotic points
		\filldraw[color=black, fill=white, thick](LM) circle (0.12);
		\filldraw[color=black, fill=white, thick](RM) circle (0.12);
	\end{tikzpicture}
	\caption{Cold BH}
	\label{subfig_cold}
\end{subfigure}
\caption{Penrose diagrams for extremal Reissner--Nordstr\"om--de Sitter black holes. (a) extremal Schwarzschild--dS black hole~\cite{lake_effects_1977,podolsky_structure_1999}, (b) charged Nariai black hole, (c) ultracold black hole, and (d) cold black hole~\cite{cardoso_nariai_2004}. Non-trivial asymptotic points are denoted by a circle.}\label{fig_Penrose}
\end{figure}
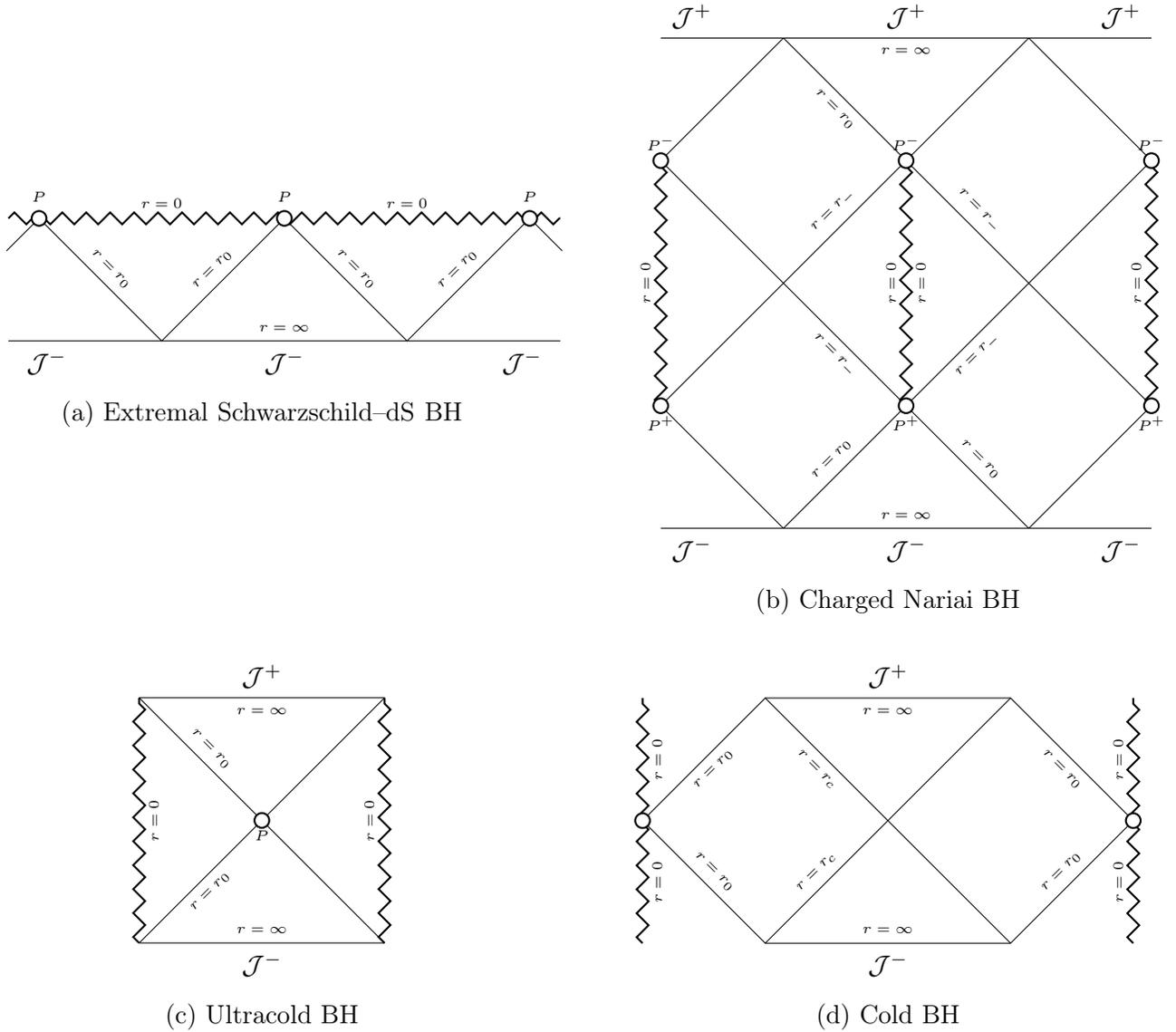

\pagebreak[1]
The special case $r_0^2\Lambda=1/2$ is known as the ultracold black hole, for which the three horizons coincide. The Penrose diagram is shown in Fig. \ref{subfig_UC} \cite{cardoso_nariai_2004}. The near-horizon geometry is $\mathbb{R}^{1,1}\times$S$^2$. One can check that, similarly to the `large' horizon case, timelike geodesics with $t=\const$ end on an asymptotic point $P$. One can thus interpret the null hypersurfaces $r=r_0$ as event horizons with respect to such observers.

\section{Transverse deformations of charged, static extremal horizons} \label{sec_NH}

\subsection{Near-horizon geometry and transverse deformations}

We will assume that $(M, g, \mathcal{F})$ is an analytic solution to the Einstein--Maxwell equations (\ref{eq_EinsteinMaxwell}) containing a degenerate Killing horizon $\mathcal{H}$ of a complete Killing vector field $K$ that also preserves the Maxwell field. In a neighbourhood of any smooth Killing horizon one can define Gaussian null coordinates \cite{moncrief_symmetries_1983} in which the metric and the Maxwell field take the form
\begin{align}
	g &= \phi\td v^2 + 2\td v\td \rho + 2\beta_a\td x^a\td v +\mu_{ab}\td x^a\td x^b\;, \label{eq_GNmetric}\\ 
	\mathcal{F} &= \Psi \td v\wedge \td \rho + Z_a\td \rho \wedge \td x^a +V_a\td v\wedge \td x^a + \frac{1}{2}B_{ab}\td x^a\wedge \td x^b\;, \label{eq_GNMaxwell}
\end{align} 
where $K=\partial_v$, the horizon is at $\rho=0$, and $(x^a)$ is a chart on the spatial cross-sections of the horizon $S$, which we assume to be compact. This coordinate system is unique up to the choice of the spatial cross-section $S$ and coordinates on $S$. One can show that $\phi$, $\beta_a$ and $V_a$ vanish at the horizon $\rho=0$ \cite{kunduri_classification_2013}, furthermore, extremality requires that $\partial_\rho\phi|_{\rho=0}=0$. Thus, we can introduce
\begin{align}
	\phi =: \rho^2 F\;, \qquad \beta_a =: \rho h_a\;, \qquad V_a =: \rho W_a\;.\label{eq_extremalGN}
\end{align} 
In these coordinates, the solution data decomposes into functions $F, \Psi$, one-forms $h=h_a\td x^a, Z=Z_a\td x^a, W=W_a\td x^a$, a two-form $B = \frac{1}{2}B_{ab}\td x^a\wedge \td x^b$ and a metric $\mu=\mu_{ab}\td x^a\td x^b$, which are defined on codimension-2 surfaces $S_{v\rho}$ of constant $v$ and $\rho$, and the horizon cross-section $S$ can be identified with $S_{v,0}$. 

It is known that extremal horizons admit a well-defined near-horizon limit, which is introduced as follows. We define the scaling diffeomorphism for $\epsilon >0$ as $\phi_\epsilon:(v, \rho, x^a)\mapsto(v/\epsilon, \epsilon \rho, x^a)$. Then the near-horizon geometry is defined as 
\begin{equation}
	g_{NH} := \lim_{\epsilon \to 0}\phi_\epsilon ^*g  = \rho^2\mathring F\td v^2 + 2\td v\td \rho + 2\rho \mathring h\td v + \mathring\mu\;, \label{eq_NHgeometry}
\end{equation} 
where we introduced the notation that $\mathring X :=X|_{\rho=0}$. The Maxwell field also admits a well-defined near-horizon limit
\begin{equation}
	\mathcal{F}_{NH} := \lim_{\epsilon \to 0}\phi_\epsilon ^*\mathcal{F} = \mathring\Psi \td v\wedge \td \rho + \rho\td v\wedge \mathring W +\mathring B\;. \label{eq_NHMaxwell}
\end{equation}
The near-horizon data consists of two functions $\mathring F, \mathring \Psi$, two one-forms $\mathring h, \mathring W$, a two-form $\mathring B$ and a metric $\mathring \mu$ on $S$, and contain the solution data evaluated at $\rho=0$, except for the one-form $Z$ which drops out in the limit. Note, that $(g_{NH}, \mathcal{F}_{NH})$ is also a solution of the Einstein--Maxwell equations (\ref{eq_EinsteinMaxwell}). 

The notion of {\it transverse deformations} was introduced in \cite{li_transverse_2016}, with the first-order deformation defined as $\left.\frac{\td}{\td\epsilon}\phi_\epsilon^*g\right|_{\epsilon=0}$, and similarly for the Maxwell field \cite{li_electrovacuum_2019}. This contains first derivatives of the solution data at the horizon, and in general, the $n$-th order transverse deformation contains information about the $n$-th derivative of the solution data, that is,
\begin{align}
	&\phi^{(n+2)}:=\partial_\rho^{n+2} \phi|_{\rho=0}=(n+2)(n+1)F^{(n)}\;,\qquad \beta^{(n+1)}:=\partial_\rho^n \beta_a|_{\rho=0}=(n+1)h_a^{(n)}\;,\nonumber\\
	&\qquad\mu_{ab}^{(n)}:=\partial_\rho^n \mu_{ab}|_{\rho=0}\;,\qquad \Psi^{(n)}:=\partial_\rho^n \Psi|_{\rho=0}\;,\qquad W_a^{(n)}:=\partial_\rho^n W_a|_{\rho=0}\;, \\
	&\qquad\qquad\qquad Z_a^{(n-1)}:=\partial_\rho^{n-1} Z_a|_{\rho=0}\;,\qquad B_{ab}^{(n)}:=\partial_\rho^n B_{ab}|_{\rho=0}\;.\nonumber
\end{align}
(Note the shift in the order for $\phi, \beta, Z$.) In the following, for any function $X$ and $n\ge 0$, we will use the notation $X^{(n)}:=\partial_\rho^nX|_{\rho=0}$, so $X^{(0)}=\mathring X$. 

The transverse deformations are not uniquely defined, as the Gaussian null coordinates admit a gauge freedom corresponding to shifting the spatial cross-section of the horizon. This leaves the near-horizon geometry unchanged (analogously to supertranslations in asymptotic symmetries), but changes the transverse deformation data. For example, under a gauge transformation generated by an arbitrary function $f$ on $S$, the first order horizon metric deformation changes as \cite{li_electrovacuum_2019}
\begin{equation}
	%F^{(1)}&\to F^{(1)} +\frac{1}{2}(\mathring\nabla^af)(\mathring\nabla_a \mathring F-\mathring h_a\mathring F)\;, \\
	%h^{(1)}_a&\to h^{(1)}_a -\frac{1}{4}\left[2\mathring F \mathring \mu_{ab} f+\mathring\nabla_a\mathring h_b+\mathring h_a\mathring h_b-2\mathring\nabla_b\mathring h_a - \mathring h_b\mathring\nabla_a\right]\mathring\nabla^b f\;,\\
	\mu^{(1)}_{ab}\to\mu^{(1)}_{ab} +\mathring\nabla_a\mathring\nabla_b f - \mathring h_{(a}\mathring\nabla_{b)} f\;.\label{eq_mu_gauge}
	%\Psi^{(1)}&\to\Psi^{(1)} + \left(\mathring\nabla_a\mathring\Psi -2\mathring h_a\mathring\Psi\right)\mathring\nabla^a f\;,\\
	%\mathring Z_a&\to\mathring Z_a+\frac{1}{2}\left(\mathring\Psi\mathring\nabla_a f-\mathring B_{ab}\mathring\nabla^b f\right)
\end{equation}
(Transformation of the complete first order deformation data can be found in \cite{li_electrovacuum_2019}.) One may fix this gauge freedom to all orders by imposing a gauge condition on the first order deformation.

\subsection{Transverse deformations to extremal RN--dS horizons}
The near-horizon data for the extremal RN--dS solutions is given by 
\begin{align}
	\mathring F =\frac{2\Lambda r_0^2-1}{r_0^2}\;, \qquad \mathring\Psi = \frac{Q}{r_0^2}=\pm \frac{1}{r_0}\sqrt{1-\Lambda r_0^2}\;,\nonumber\\
	\qquad \mathring h=\mathring W=0\;, \qquad \mathring B=0\;,\qquad \mathring\mu = r_0^2\td \Omega_2^2\;,\label{eq_NHdata}
\end{align}
with $S$ being a two-sphere and $0<\Lambda r_0^2\le 1$. Here, we take the near-horizon data of electrically charged black holes, but any static near-horizon data is related to (\ref{eq_NHdata}) by an electromagnetic duality transformation. We emphasize that we do {\it not} assume that the full solution is purely electric. Our goal is to compute all higher order deformations to this near-horizon geometry. Uniqueness of the first order deformations was proven in \cite{kolanowski_towards_2021}, which we include here in our notation for completeness.
\begin{proposition}\label{prop_1}
	Consider a solution of (\ref{eq_EinsteinMaxwell}) with $\Lambda>0$ containing a degenerate Killing horizon with the near-horizon geometry of the extremal Reissner--Nordstr\"om--de Sitter horizon (\ref{eq_NHdata}), with the Killing field generating the horizon also preserving the Maxwell field. Then the first order transverse deformations to the degenerate horizon are gauge-equivalent to 
	\begin{align}
		\phi^{(3)} = \frac{2C}{r_0^2}(3-4\Lambda r_0^2)\;, \qquad\qquad \beta^{(2)}=0\;, \qquad \mu^{(1)} = C\mathring\mu \;,\nonumber \\
		\Psi^{(1)}=-C\mathring\Psi\;,\qquad Z^{(0)}=0\;, \qquad W^{(1)}=0\;,\qquad B^{(1)}=0\;, \label{eq_deformation1}
	\end{align} 
	for some $C\in \mathbb{R}$.
\end{proposition}
\begin{proof}
	From the transformation of $\mu^{(1)}$ under gauge transformations (\ref{eq_mu_gauge}), it is clear that we may fix the gauge freedom of the transverse deformations by requiring that 
	\begin{equation}
		\mathring\mu^{ab} \mu^{(1)}_{ab} = 2C \label{eq_trace1}
	\end{equation} 
	for some constant $C$. This is possible due to standard existence and uniqueness theorems for Poisson's equation on $S$. 
	
	The Bianchi identities $\mathcal{B}_W^{(1)}=0$ and $\mathring{\mathcal{B}}_B=0$ (\ref{eq_BianchiW}-\ref{eq_BianchiB}) show that $W^{(1)}$ and $B^{(1)}$ are algebraically determined by $\Psi^{(1)}$ and $Z^{(0)}$ as
	\begin{align}
		W^{(1)} &= \frac{1}{2} \td \Psi^{(1)}\;, \label{eq_W1}\\
		B^{(1)} &= \td Z^{(0)}\;. \label{eq_B1}
	\end{align}
	Using this, the Einstein--Maxwell equations $\mathring{\mathcal{E}}_{ra}=0$, $\mathcal{E}^{(1)}_{rv}=0$, $\mathring{\mathcal{M}}^v=0$ (see (\ref{eq_phi}-\ref{eq_Z})) respectively yield
	\begin{align}
		\beta^{(2)}_a&=-\mathring\nabla^c\mu_{ac}^{(1)}+4\mathring\Psi Z_a^{(0)}\;, \label{eq_beta1}\\
		\phi^{(3)}&=\frac{2C}{r_0^2}(1-2\Lambda r_0^2)-\mathring\nabla\cdot\beta^{(2)}-4\mathring\Psi\Psi^{(1)}\;,\label{eq_phi1}\\
		\Psi^{(1)}&=\mathring\nabla\cdot Z^{(0)}-C\mathring\Psi\;, \label{eq_psi1}
	\end{align}
	and thus $\beta^{(2)},\phi^{(3)},\Psi^{(1)}$ are algebraically determined by $\mu^{(1)}$ and $Z^{(0)}$.

	Substituting (\ref{eq_W1}-\ref{eq_psi1}) into the Einstein--Maxwell equations $\left(\mathcal{M}^a{}\right)^{(1)}=0$ and $\mathcal{E}^{(1)}_{ab}=0$ yields a system of coupled PDEs for $\mu^{(1)}$ and $Z^{(0)}$ on $S$.
	\begin{align}
		&\qquad\qquad-\mathring\nabla^2Z^{(0)}_a+\frac{3}{r_0^2}Z^{(0)}_a-\mathring\Psi\mathring\nabla^b\tilde\mu^{(1)}_{ab}=0\;, \label{PDE_Z1} \\
		&\left(-\frac{1}{2}\mathring\nabla^2+2\Lambda\right)\tilde\mu^{(1)}_{ab} + 4\mathring\Psi\left[\mathring\nabla_{(a}Z^{(0)}_{b)}-\frac{1}{2}\left(\mathring\nabla\cdot Z^{(0)}\right)\mathring\mu_{ab}\right]=0\;.\label{PDE_mu1}
	\end{align}
	where we defined the traceless tensor $\tilde\mu^{(1)}_{ab}:=\mu^{(1)}_{ab}-C\mathring\mu_{ab}$. For (\ref{PDE_mu1}) we used that the first derivative of the Ricci tensor $\mathcal{R}_{ab}$ of $(S_{v,\rho}, \mu)$ is given by
	\begin{equation}
		\mathcal{R}^{(1)}_{ab} = \mathring\Delta_L\mu^{(1)}_{ab} + \mathring\nabla_{(a}v^{(1)}_{b)}
		\;,\label{eq_ricci_pert}
	\end{equation}
	where $\mathring\Delta_L$ is the Lichnerowicz operator of $(S, \mathring \mu)$, which is explicitly defined as
	\begin{equation}
		\mathring\Delta_L\mu^{(1)}_{ab}:= -\frac{1}{2}\mathring\nabla^2\mu^{(1)}_{ab} - \mathcal{\mathring{R}}_a{}^c{}_b{}^d\mu^{(1)}_{cd} 
		+ \mathcal{\mathring{R}}_{(a}^c\mu^{(1)}_{b)c}\;,
		\label{eq_def_lichnerowicz}
	\end{equation}
	and 
	\begin{equation}
		v^{(1)}_b:= \mathring\nabla^c\mu^{(1)}_{bc} - \frac{1}{2}\mathring\nabla_b\left(\mathring\mu^{cd}\mu^{(1)}_{cd}\right)\;. 
		\label{eq_def_v}
	\end{equation}
	
	For the special case $\mathring\Psi=0$, (\ref{PDE_Z1}-\ref{PDE_mu1}) decouple and uniqueness of the trivial solution follows from positivity of $-\mathring\nabla^2$. To solve (\ref{PDE_Z1}-\ref{PDE_mu1}) for general $\mathring\Psi\neq0$, we use the fact that one-forms on $S$ admit an orthogonal basis defined by
	\begin{align}
		V^{\{lm+\}} := \td Y^{lm}, \qquad V^{\{lm-\}}:= \star_2 V^{\{lm+\}}\;,\label{eq_basis1form}
	\end{align}
	where $Y^{lm}$ are spherical harmonics satisfying $-\mathring\nabla^2Y^{lm} = \frac{l(l+1)}{r_0^2}Y^{lm}$ for integers $l\ge1$ and $-l\le m\le l$, and $\star_2$ is the Hodge operator on $(S,\mathring\mu)$. Similarly, for traceless symmetric 2-tensors, we use the basis defined by\footnote{Note that $\mathring\nabla\cdot V^{\{lm-\}}=0$.}
	\begin{align}
		T_{ab}^{\{lm+\}}:=\mathring\nabla_{(a}V^{\{lm+\}}_{b)} - \frac{1}{2}\mathring\mu_{ab}\mathring\nabla\cdot V^{\{lm+\}}\;, \qquad T_{ab}^{\{lm-\}}:=\mathring\nabla_{(a}V^{\{lm-\}}_{b)}, \label{eq_basis2tensor}
	\end{align}
	for $l\ge2$ (the analogous 2-tensors vanish for $l=1$). This basis diagonalises the Laplacian. Indeed, one can show that 
	\begin{align}
		-\mathring\nabla^2 V^{\{lm\pm\}} = \frac{l(l+1)-1}{r_0^2}V^{\{lm\pm\}}\;,\label{eq_eigen_1form}
	\end{align}
	and 
	\begin{align}
		-\mathring\nabla^2 T^{\{lm\pm\}} = \frac{l(l+1)-4}{r_0^2}T^{\{lm\pm\}}\;.\label{eq_eigen_STT}
	\end{align}
	Furthermore, one can easily check that $\mathring\nabla\cdot T^{\{lm\pm\}}=\frac{2-l(l+1)}{2r_0^2}V^{\{lm\pm\}}$. Since indices $m, \pm$ do not play any role, we suppress them in the following.
	
	Let us expand $Z^{(0)}$ and $\tilde\mu^{(1)}$ in the basis (\ref{eq_basis1form}-\ref{eq_basis2tensor}). Let $\Pi^{\{l\}}$ denote the projection\footnote{
		With a slight abuse of notation we denote the projection of 1-forms and 2-tensors by the same symbol.
	} of the tensors to the one-dimensional subspace of (\ref{eq_basis1form}-\ref{eq_basis2tensor}) with a given $l$ (and given $m, \pm$). Since the operators appearing in (\ref{PDE_Z1}-\ref{PDE_mu1}) commute with $\Pi^{\{l\}}$, we may restrict (\ref{PDE_Z1}-\ref{PDE_mu1}) to a given $l$. For $l=1$ the corresponding 2-tensor $\Pi^{\{l=1\}}\tilde\mu^{(1)}$ vanishes, thus (\ref{PDE_Z1}) automatically yields $\Pi^{\{l=1\}}Z^{(0)}=0$ since $-\mathring\nabla^2$ is a positive operator on $S$. For $l\ge 2$, if $\Pi^{\{l\}}Z^{(0)}=0 \implies \Pi^{\{l\}}\tilde\mu^{(1)}=0$ by (\ref{PDE_Z1}), as there are no traceless transverse tensors on $S^2$ \cite{katona_uniqueness_2023}. Thus, we will consider $\Pi^{\{l\}}Z^{(0)}\neq0$, in which case $\Pi^{\{l\}}\tilde\mu^{(1)}$ must take the form\footnote{This follows from the fact that $\Pi^{\{l\}}$ projects onto a one-dimensional subspace (recall that indices $m,\pm$ are suppressed).}
	\begin{align}
		&\Pi^{\{l\}}\tilde\mu^{(1)}_{ab} = \frac{\alpha^{\{l\}}}{\mathring\Psi}\left[\mathring\nabla_{(a}\Pi^{\{l\}}Z^{(0)}_{b)}-\frac{1}{2}\left(\mathring\nabla\cdot \Pi^{\{l\}}Z^{(0)}\right)\mathring\mu_{ab}\right]\label{eq_muansatz}
	\end{align}
	with some constant $\alpha^{\{l\}}$, and a factor of $\mathring\Psi^{-1}$ is included for convenience. Equations (\ref{PDE_Z1}-\ref{PDE_mu1}) reduce to algebraic system of equations for $(\alpha^{\{l\}}, l)$. For $l$ we get $l=0$ or $l(l+1) = 6-8\Lambda r_0^2$. Thus, we conclude that for any $l\ge2$ the unique solution is the trivial one: $Z^{(0)}=0$ and $\tilde\mu^{(1)}=0$. The rest of (\ref{eq_deformation1}) follows from (\ref{eq_W1}-\ref{eq_psi1}).
\end{proof}

Let us now generalise Proposition \ref{prop_1} to higher order deformations, assuming that the lower order deformations agree with those of extremal Reissner--Nordstr\"om--dS. For reference, we list the $k$-th order deformation data of extremal RN--dS here:
\begin{align}
	&\phi^{(k+2)} =\frac{1}{r_0^2}\left[\frac{2\Lambda r_0^2}{3}\delta_{0k}-(1-\Lambda r_0^2)(k+3)!+2\left(1-\frac{2\Lambda r_0^2}{3}\right)(k+2)!\right]\left(-\frac{C}{2}\right)^k\;, \nonumber\\
	&\Psi^{(k)} = \mathring\Psi(k+1)!\left(-\frac{C}{2}\right)^k\;,\nonumber\\
	&\mu^{(k)} =\left\{
		\begin{array}{ll}
		\mathring\mu_{ab} \;,  & \text{if } k =0 \\
			C\mathring\mu_{ab}\;, &  \text{if } k = 1 \\
			\frac{C^2}{2}\mathring\mu_{ab}\;, & \text{if } k = 2\\
			0\;, & \text{if } 3\leq k ,
		\end{array}\right. \label{eq_deformation_k}\\
	&\beta^{(k+1)}=0 \;, \qquad Z^{(k-1)}=0 \;, \qquad W^{(k)}=0\;, \qquad B^{(k)} = 0 \;.\nonumber
\end{align}

\begin{proposition}\label{prop_2}
	Consider a spacetime with a degenerate horizon as in Proposition \ref{prop_1}. Assume furthermore that its $k$-th order transverse deformation agrees with (\ref{eq_deformation_k}) for all $1\le k \le n-1$. Then the $n$-th order transverse deformation is unique  and agrees with (\ref{eq_deformation_k}) given that $r_0^2\Lambda\in (0, 1]\setminus E_{\nu=n(n+1)}$, where 
	\begin{equation}
		E_\nu := \left\{\frac{1}{2}-\frac{\nu\lambda\pm \sqrt{2\lambda\nu(\nu+\lambda-2)}}{2\nu(\nu-2)}\Bigg|\lambda = l(l+1) \text{ for } l\in\mathbb{Z}^+\right\}\setminus \left\{1/2\right\}\;.\label{eq_Edef}
	\end{equation}
\end{proposition}
\begin{proof}
	From the assumption on the lower order terms it follows that
	\begin{align}
		&\mu_{ab} = \left(1+\frac{C\rho}{2}\right)^2\mathring\mu_{ab} + \ord\left(\rho^n\right)\;, \qquad \mu^{ab} = \left(1+\frac{C\rho}{2}\right)^{-2}\mathring\mu^{ab} + \ord\left(\rho^n\right)\;,\nonumber\\
		&\Psi = \frac{\mathring\Psi}{\left(1+\frac{C\rho}{2}\right)^2}+\ord(\rho^n)\;, \qquad \beta = \ord\left(\rho^{n+1}\right)\;,\nonumber\\
		&Z =\ord\left(\rho^{n-1}\right) \;, \qquad W = \ord\left(\rho^n\right)\;, \qquad B = \ord\left(\rho^n\right)\;.\label{assumption_n}
	\end{align}
	Using (\ref{assumption_n}) for $Z$, the equation $\mathcal{E}_{rr}^{(n-2)}=0$ can be written as 
	\begin{equation}
		-\frac{1}{2}\mathring\mu^{ab}\mu_{ab}^{(n)} + \left[\frac{1}{2}\mathring\mu^{ab}\mu_{ab}^{(n)}+ \frac{1}{4}\left(\mu^{ac}\mu^{bd}\dot\mu_{ab}\dot\mu_{cd}-2\mu^{ab}\ddot\mu_{ab}\right)^{(n-2)}\right]=0\;,
	\end{equation}
	where dot denotes the derivative with respect to $\rho$ without evaluating at $\rho=0$. Notice that the square bracket contains derivatives that are of lower order (the $n$-th derivative terms cancel), which can be evaluated explicitly using (\ref{assumption_n}). Thus, the unique solution for the trace of $\mu^{(n)}$ must agree with that of (\ref{eq_deformation_k}), in particular, it must be constant:
	\begin{equation}
		\mathring\tr\mu^{(n)}:=\mathring\mu^{ab}\mu_{ab}^{(n)}=\left\{\begin{array}{ll}
			C^2\;,  	& \text{if } n=2 \\
			0\;, 				&  \text{if } n >2 \;.
			\end{array}\right.\label{eq_trace_n}
	\end{equation}

	We first show that given $\mu^{(n)}$ and $Z^{(n-1)}$, the rest of the $n$-th order data ($\phi^{(n+2)}$, $\Psi^{(n)}$, $\beta^{(n+1)}, W^{(n)}$, $B^{(n)}$) is uniquely determined by algebraic equations. Using (\ref{assumption_n}) and (\ref{eq_trace_n}), the equation $\mathcal{E}_{ra}^{(n-1)}=0$ simplifies to 
	\begin{equation}
		\beta^{(n+1)}_a = -\mathring\nabla^b\mu_{ab}^{(n)}+4\mathring\Psi Z_a^{(n-1)}\;, \label{eq_beta_n}
	\end{equation}
	so $\beta^{(n+1)}$ is algebraically determined by $\mu^{(n)}$ and $Z^{(n-1)}$. The Maxwell--equation $\left(\mathcal{M}^{v}\right)^{(n-1)}=0$ can be written as 
	\begin{equation}
		\Psi^{(n)} = \mathring\nabla\cdot Z^{(n-1)} -\frac{1}{2} \mathring\Psi\mathring\mu^{ab}\mu_{ab}^{(n)}-\frac{1}{2}\left[(\Psi\mu^{ab}\dot\mu_{ab})^{(n-1)}-\mathring\Psi\mathring\mu^{ab}\mu_{ab}^{(n)}\right]\;,\label{eq_psi_n}
	\end{equation}
	where we used (\ref{assumption_n}) for $\beta$ and $Z$. The square bracket in (\ref{eq_psi_n}) only contains lower order data (terms with the highest derivative of $\mu$ cancel), and the second term of (\ref{eq_psi_n}) can be evaluated by (\ref{eq_trace_n}), thus $\Psi^{(n)}$ is determined algebraically in terms of $Z^{(n-1)}$. The Einstein--equation $\mathcal{E}_{vr}^{(n)}=0$ yields by (\ref{assumption_n}) that
	\begin{equation}
		\phi^{(n+2)} = -\frac{n}{2}\mathring\mu^{ab}\mu_{ab}^{(n)}\phi^{(2)} -\mathring\nabla\cdot\beta^{(n+1)} - 2\left(\Psi^2\right)^{(n)} - \frac{1}{2}\left[\left(\mu^{ab}\dot\mu_{ab}\dot\phi\right)^{(n)}-n\mathring\mu^{ab}\mu_{ab}^{(n)}\phi^{(2)}\right]\;,
	\end{equation}
	where the square bracket contains only lower order data, and the remaining terms can be computed by (\ref{eq_trace_n}-\ref{eq_psi_n}) and (\ref{assumption_n}). Finally, derivatives of the Bianchi identities (\ref{eq_BianchiW}-\ref{eq_BianchiB}) express $W^{(n)}, B^{(n)}$ algebraically in terms of $\Psi^{(n)}, Z^{(n-1)}$ as
	\begin{align}
		W^{(n)}=\frac{1}{n+1}\td \Psi^{(n)}\;, \qquad B^{(n)}=\td Z^{(n-1)}\;.\label{eq_B_n}
	\end{align}

	Now we turn to the equations for $\mu^{(n)}$ and $Z^{(n-1)}$. The Einstein--equation $\mathcal{E}_{ab}^{(n)}=0$ can be written as 
	\begin{align}
		&\mathring\nabla_{(a}\beta^{(n+1)}_{b)} + \frac{1}{2}\binom{n+1}{2}\phi^{(2)}\mu_{ab}^{(n)} + \mathcal{R}_{ab}^{(n)}-\Lambda\mu_{ab}^{(n)}-\mathring\Psi^2\mu_{ab}^{(n)} -2\mathring\Psi\Psi^{(n)}\mathring\mu_{ab}= \nonumber \\
		&\qquad= -\frac{1}{4}\left(\mu^{cd}\dot\mu_{cd}\phi\dot\mu_{ab}\right)^{(n)}-\frac{1}{2}\left[\left(\phi\dot\mu_{ab}\right)^{(n+1)}-\binom{n+1}{2}\phi^{(2)}\mu_{ab}^{(n)}\right]\nonumber\\
		&\qquad\qquad +\frac{1}{2}\left(\phi\mu^{cd}\dot\mu_{ac}\dot\mu_{bd}\right)^{(n)} + \left[\left(\Psi^2\mu_{ab}\right)^{(n)}-\mathring\Psi^2\mu_{ab}^{(n)}-2\mathring\Psi\Psi^{(n)}\mathring\mu_{ab}\right]\;. \label{eq_mu_n}
	\end{align}
	Again, the highest order derivatives in the square brackets cancel, and thus the right-hand side of (\ref{eq_mu_n}) contains only lower order derivatives (recall that the first derivative of $\phi$ vanish at $\rho=0$). Equation (\ref{eq_mu_n}) can be simplified as follows. In Proposition 2 of \cite{katona_uniqueness_2023} it has been shown that the $n$-th order variations of the Ricci tensor $\mathcal{R}_{ab}^{(n)}$ of (\ref{assumption_n}) can be evaluated as 
	\begin{equation}
		\mathcal{R}_{ab}^{(n)}  = 
		\mathring\Delta_L\mu^{(n)}_{ab} + \mathring\nabla_{(a}v^{(n)}_{b)}\;, \label{eq_R_npert}
	\end{equation}
	where $v^{(n)}_a=\mathring\nabla^b\mu^{(n)}_{ba}-\frac{1}{2} \mathring\nabla_a (\mathring\mu^{cd}\mu^{(n)}_{cd})$. Using (\ref{eq_R_npert}) and (\ref{eq_trace_n}-\ref{eq_psi_n}), we can eliminate all $n$-th order fields except for $Z^{(n-1)}, \mu^{(n)}$. As for the first order deformations, it is useful to change variables $\mu^{(n)} \to \tilde\mu^{(n)}:= \mu^{(n)}-\frac{1}{2}\mathring\mu\mathring\tr\mu^{(n)}$. Since we know $Z^{(n-1)}=0, \tilde\mu^{(n)} = 0$ is a solution (corresponding to our model solution (\ref{eq_deformation_k})), all the lower order terms, which do not contain $Z^{(n-1)}$ or $\tilde\mu^{(n)}$, must cancel. Thus, we obtain for (\ref{eq_mu_n})
	\begin{align}
		&-\mathring\nabla^2\tilde\mu_{ab}^{(n)}=-\frac{2}{r_0^2}\left[1+\binom{n+1}{2}\left(2\Lambda r_0^2-1\right)\right]\tilde\mu_{ab}^{(n)}-8\mathring\Psi\left(\mathring\nabla_{(a}Z_{b)}^{(n-1)}-\frac{1}{2}\mathring\mu_{ab}\mathring\nabla\cdot Z^{(n-1)}\right)\;. \label{eq_mu_n_final}
	\end{align}
	Finally, the Maxwell--equation $\left(\mathcal{M}^{a}\right)^{(n)}=0$, after substituting in (\ref{eq_B_n}), yields
	\begin{equation}
		-\mathring\nabla^2Z^{(n-1)}_a +\mathring\nabla^b \mathring\nabla_aZ^{(n-1)}_b + \mathring\Psi\beta_a^{(n+1)} +\binom{n+1}{2}\phi^{(2)} Z^{(n-1)}_a-\mathring\nabla_a\Psi^{(n)}=0. \label{eq_Z_n}
	\end{equation}
	Substituting in (\ref{eq_beta_n}), (\ref{eq_psi_n}) and (\ref{eq_deformation_k}) for $\beta^{(n+1)}$, $\Psi^{(n)}$, and $\phi^{(2)}$, respectively, we obtain
	\begin{equation}
		-\mathring\nabla^2Z^{(n-1)}_a =-\frac{1}{r_0^2}\left[5-4\Lambda r_0^2+n(n+1)\left(2\Lambda r_0^2-1\right)\right]Z^{(n-1)}_a+ \mathring\Psi\mathring\nabla^b\tilde\mu_{ab}^{(n)}\label{eq_Z_n_final}\;,
	\end{equation}
	where we also used the identity on $(S, \mathring\mu)$ that  $\left(\mathring\nabla^b\mathring\nabla_a-\mathring\nabla_a\mathring\nabla^b\right)Z_b =\frac{1}{r_0^2} Z_a$. The claim that the $n$-th order deformation is unique and given by (\ref{eq_deformation_k}) is equivalent to (\ref{eq_mu_n_final}) and (\ref{eq_Z_n_final}) admitting only the trivial solution.

	In the uncharged case, when $\mathring\Psi=0$ (or equivalently $r_0^2\Lambda=1$), the two equations decouple into two eigenvalue equations of positive operators $-\mathring\nabla^2$. Notice that in such a case, the eigenvalues are negative for all $n$, thus the only solution is the trivial solution, proving the claim for $\mathring\Psi=0$. 

	For the charged case $\mathring\Psi\neq0$, we can solve the coupled system of equations (\ref{eq_mu_n_final}), (\ref{eq_Z_n_final}) following the same strategy as in Proposition \ref{prop_1}, expanding $ Z^{(n-1)}, \tilde\mu^{(n)}$ in the spherical harmonic basis (\ref{eq_basis1form}-\ref{eq_basis2tensor}). Again, the operators in (\ref{eq_mu_n_final}) and (\ref{eq_Z_n_final}) commute with projections $\Pi^{\{l\}}$, thus we may consider solutions $(\Pi^{\{l\}}\tilde\mu^{(n)}, \Pi^{\{l\}} Z^{(n-1)})$ for each $l$ separately. Since there are no traceless divergence-free 2-tensors on $(S,\mathring\mu)$, if $\Pi^{\{l\}}Z^{(n-1)}=0\implies \Pi^{\{l\}}\tilde\mu^{(n)}=0$ by (\ref{eq_Z_n_final}), thus we only consider the case $\Pi^{\{l\}}Z^{(n-1)}\neq0$ for some $l\ge1$.

	Let us first look at the projection to the $l=1$ sector: $\Pi^{\{l=1\}}Z^{(n-1)}$ and $\Pi^{\{l=1\}}\tilde\mu^{(n)}$. Since (\ref{eq_basis2tensor}) vanishes for $l=1$, we automatically have $\Pi^{\{l=1\}}\tilde\mu^{(n)}=0$. Using (\ref{eq_eigen_1form}), equation (\ref{eq_Z_n_final}) reduces to the algebraic equation
	\begin{equation}
		\left[6-\nu +2 (\nu-2)r_0^2\Lambda\right]\Pi^{\{l=1\}}Z^{(n-1)}=0\;, \label{eq_algebraic_Z_n_1}
	\end{equation}
	where we used the notation $\nu:=n(n+1)$. Eq. (\ref{eq_algebraic_Z_n_1}) admits a non-trivial solution $\Pi^{\{l=1\}}Z^{(n-1)}\neq0$ if and only if 
	\begin{equation}
		r_0^2\Lambda  = \frac{\nu-6}{2(\nu-2)}\;.\label{eq_l1}
	\end{equation}

	For $l\ge 2$, following the arguments of Proposition \ref{prop_1}, we write 
	\begin{equation}
		\Pi^{\{l\}}\tilde\mu^{(n)}_{ab} = \frac{\alpha^{\{l\}}}{\mathring\Psi}\left[\mathring\nabla_{(a}\Pi^{\{l\}}Z^{(n-1)}_{b)}-\frac{1}{2}\left(\mathring\nabla\cdot \Pi^{\{l\}}Z^{(n-1)}\right)\mathring\mu_{ab}\right] \label{eq_muansatz_n}
	\end{equation}
	with some constants $\alpha^{\{l\}}$. Using (\ref{eq_eigen_1form}-\ref{eq_eigen_STT}) and the divergence formula for (\ref{eq_basis2tensor}) we obtain a set of algebraic equations for any non-trivial solution of (\ref{eq_mu_n_final}) and (\ref{eq_Z_n_final}) for a given $l\ge2$:
	\begin{equation}
		\left.
	\begin{array}{rl}
		\alpha^{\{l\}}(\lambda-4) &= -\alpha^{\{l\}}(2+2\nu r_0^2\Lambda -\nu) -8(1-r_0^2\Lambda)\\
		\lambda-1 &= -5-2\nu r_0^2\Lambda +\nu +4r_0^2\Lambda +\frac{2-\lambda}{2}\alpha^{\{l\}}
	\end{array}\right\}\;\label{eq_algebraic_n}
	\end{equation}
	with $\nu:=n(n+1)$ and $\lambda:=l(l+1)$, as before. For a given $\nu$ and $\lambda$, (\ref{eq_algebraic_n}) can be solved for $r_0^2\Lambda$ and  $\alpha^{\{l\}}$, for which we get
	\begin{align}
		r_0^2\Lambda &= \frac{1}{2}-\frac{\nu\lambda\pm \sqrt{2\lambda\nu(\nu+\lambda-2)}}{2\nu(\nu-2)}\;,\label{eq_sol_Lambda}\\
		\alpha^{\{l\}} &= -2\frac{2\nu\mp\sqrt{2\lambda\nu(\nu+\lambda-2)} }{\nu(\lambda-2)}\;.\label{eq_sol_alpha}
	\end{align}
	Note that we obtain (\ref{eq_l1}) from (\ref{eq_sol_Lambda}) with $\lambda=2$ and taking the upper sign. 
	
	We have derived that if a non-trivial solution to (\ref{eq_mu_n_final}) and (\ref{eq_Z_n_final}) exists, $r_0^2\Lambda$ must take the form (\ref{eq_l1}), or (\ref{eq_sol_Lambda}) for some $l\ge 2$. Taking the contrapositive, if $r_0^2\Lambda$ is not of the form (\ref{eq_l1}), or (\ref{eq_sol_Lambda}) for any $l\ge 2$, the only solution to (\ref{eq_mu_n_final}) and (\ref{eq_Z_n_final}) is the trivial one $\tilde \mu^{(n)}=0$, $Z^{(n-1)}=0$. In such a case, the rest of the $n$-th order transverse deformation data is uniquely determined by the algebraic equations (\ref{eq_beta_n}-\ref{eq_B_n}). Conversely, if (\ref{eq_l1}) holds or there exists an integer $l\ge 2$ for a given $\nu$ such that (\ref{eq_sol_Lambda}) holds, equations (\ref{eq_mu_n_final}) and (\ref{eq_Z_n_final}) admit a non-trivial solution of the form $\tilde \mu^{(n)}=0$ and $Z^{(n-1)}\in\operatorname{span}\{V^{\{1\}}\}$, or $Z^{(n-1)}\in\operatorname{span}\{ V^{\{l\}}\}$  and $\mu^{(n)} = \Pi^{\{l\}}\mu^{(n)}$ with (\ref{eq_muansatz_n}), respectively.

	One can easily check that the second term on the right-hand side of (\ref{eq_sol_Lambda}) vanishes for $n>1$ if and only if $\lambda=2\iff l=1$, however in that case the necessary criterion is (\ref{eq_l1}). Thus, if $r_0^2\Lambda=1/2$, (\ref{eq_mu_n_final}) and (\ref{eq_Z_n_final}) admit only the trivial solution for all $n\ge2$.
\end{proof}

Now we are ready to prove our main result.

\begin{theorem}\label{thm_uniqueness}
	Consider an analytic solution of the Einstein--Maxwell equations (\ref{eq_EinsteinMaxwell}) with $\Lambda>0$ that contains a degenerate Killing horizon with a compact cross-section $S$ and near-horizon data (\ref{eq_NHdata}) with the horizon area $A_H=4\pi r_0^2$ satisfying $A_H\Lambda\notin 4\pi E$, where $E:= \cup_{n=2}^\infty E_{\nu=n(n+1)}$ and $E_\nu$ is defined in (\ref{eq_Edef}). Let us further assume that the Maxwell field is preserved by the Killing field generating the horizon. Then the solution is given by (\ref{eq_GNmetric}-\ref{eq_GNMaxwell}) with 
	\begin{align}
		\phi &=  \left\{
			\begin{array}{ll}
				-\frac{4}{C^2r_0^2}\left[1-\frac{2\left(1-\frac{2}{3}r_0^2\Lambda\right)}{1+\frac{C\rho}{2}}-\frac{r_0^2\Lambda}{3}\left(1+\frac{C\rho}{2}\right)^2+\frac{1-r_0^2\Lambda}{\left(1+\frac{C\rho}{2}\right)^2}\right] & \text{ for } C\neq0\;,\\
				\frac{2r_0^2\Lambda-1}{r_0^2}\rho^2&\text{ for } C=0\;,
			\end{array}\right.\nonumber\\
		\mu &=  r_0^2\left(1 +\frac{C\rho}{2}\right)^2d\Omega_2^2\; ,\nonumber\\
		\Psi &=  \frac{Q}{r_0^2}\left(1+\frac{C\rho}{2}\right)^{-2}\;,\nonumber\\
		\beta&=0\;, \qquad V=0\;,\qquad Z=0\;,\qquad B=0\;,\label{eq_metric_C}
	\end{align}
	where $C\in \mathbb{R}$. If $C\neq0$, this is the extremal Reissner--Nordstr\"om--de Sitter spacetime. If $C=0$ this is its near-horizon geometry $dS_2\times S^2$ (Nariai solution) for $A_H\Lambda>2\pi$, $\mathbb{R}^{1,1}\times S^2$ for $A_H\Lambda=2\pi$, and $AdS_2\times S^2$ for $A_H\Lambda<2\pi$.
	
\end{theorem}
\begin{proof}
	By the assumption of analyticity, we may Taylor-expand the metric and Maxwell field components as $\phi = \sum_{n\ge0}\frac{\phi^{(n)}\rho^n}{n!}$, and similarly for $\beta$, $\mu$, $\Psi$, $Z$, $V=\rho W$, and $B$. The coefficients $\phi^{(n+2)}$, $\beta^{(n+1)}$, $\mu^{(n)}$, $\Psi^{(n)}$, $Z^{(n-1)}$, $W^{(n)}$, $B^{(n)}$ correspond to the $n$-th order transverse deformations. We prove uniqueness of these coefficients by induction. For the base case ($n=1$), uniqueness is established in Proposition \ref{prop_1}. The inductive step is proven in Proposition \ref{prop_2}, where we used our assumption of the horizon area $A_H$. Hence, the coefficients of metric and Maxwell field components are given by (\ref{eq_deformation_k}) to all orders. After summation, we obtain (\ref{eq_metric_C}). For $C=0$ we obtain the corresponding near-horizon geometries and the near-horizon limit of the Maxwell field. If $C\neq0$, we may rescale the coordinates as 
	\begin{align}
		v' = \frac{2}{ r_0 C}v\;, \qquad \rho' = \frac{ r_0 C}{2}\rho\;,\label{rescale}
	\end{align}
	and after shifting the radial coordinate as $r:=r_0 + \rho'$, we explicitly obtain the extremal Reissner--Nordstr\"om--de Sitter solution (\ref{eq_RNdS}).
\end{proof}
\pagebreak[3]
\noindent{\bf Remarks.}
\begin{enumerate}
	\item From the definition of $E_\nu$ in (\ref{eq_Edef}) one can show that for each $\nu=n(n+1)>2$, $E_\nu$ is bounded by $1/2$ from above, and thus $E$ is bounded by $1/2$ from above (in fact, $1/2$ is its supremum). It follows that Theorem \ref{thm_uniqueness} provides a uniqueness theorem for charged Nariai and ultracold black holes for which $A_H\Lambda\ge 2\pi$. At the same time, this method only allows us to establish a generic uniqueness result for cold black holes with $A_H\Lambda<2\pi$. It would be interesting to see what solutions, if any, the non-trivial deformations for the non-generic near-horizon data correspond to.
	\item The near-horizon geometry of the extremal Reissner--Nordstr\"om--de Sitter black hole has been shown to be unique in this theory among static near-horizon geometries with a compact cross-section \cite{kunduri_classification_2013}. That is, any such near-horizon data can be written as (\ref{eq_NHdata}) (after an electromagnetic duality transformation if necessary). It follows, that assumptions about the near-horizon data in Theorem \ref{thm_uniqueness} can be relaxed to requiring only staticity of the horizon and compactness of cross-sections with horizon area $A_H\Lambda\notin E$.
	\item Theorem \ref{thm_uniqueness} establishes the uniqueness of the (uncharged) extremal Schwarzschild--de Sitter in the Einstein--Maxwell theory with a positive cosmological constant, generalising the previous uniqueness theorem in vacuum gravity with $\Lambda>0$ \cite{katona_uniqueness_2023}.
	\item The calculations of this work can be extended to the extremal Reissner--Nordstr\"om--anti-de Sitter black holes with $\Lambda<0$, however due to the change of sign, some arguments in Proposition \ref{prop_1} and \ref{prop_2} using the positivity of the tensorial Laplacians no longer work. Thus, for a discrete set of values for $A_H\Lambda$, one gets non-trivial deformations, which allows only for a generic uniqueness result, similarly to the cold black holes. 
	\item One can also consider extremal hyperbolic Reissner--Nordstr\"om--AdS black holes, for which the geometry of the horizon cross-section is a compact hyperbolic surface of genus $g\ge 2$. The model solution is given by (\ref{eq_RNdS}) (with a hyperbolic spatial metric instead of $\td \Omega_2^2$) and $\Phi(r) = -1 -m/r -\Lambda r^2/3 +Q^2/r^2$, where $m = -2r_0(1+2\Lambda r_0^2 /3)$ and $Q^2 = -r_0^2(1+\Lambda r_0^2) $. In this case, for the equations governing the $n$-th order perturbations one finds (assuming agreement with the model solution to lower orders)
	\begin{align}
		&-\mathring\nabla^2\tilde\mu_{ab}^{(n)}=\frac{2}{r_0^2}\left[1-\binom{n+1}{2}\left(2\Lambda r_0^2+1\right)\right]\tilde\mu_{ab}^{(n)}-8\mathring\Psi\left(\mathring\nabla_{(a}Z_{b)}^{(n-1)}-\frac{1}{2}\mathring\mu_{ab}\mathring\nabla\cdot Z^{(n-1)}\right)\;\nonumber \\
		&-\mathring\nabla^2Z^{(n-1)}_a =-\frac{1}{r_0^2}\left[-5-4\Lambda r_0^2+n(n+1)\left(2\Lambda r_0^2+1\right)\right]Z^{(n-1)}_a+ \mathring\Psi\mathring\nabla^b\tilde\mu_{ab}^{(n)}\;,\label{eq_hyperbolic}
	\end{align}
	with $\mathring\Psi = \pm r_0^{-1}\sqrt{-\Lambda r_0^2 -1}$. Interestingly, one can check that the topological obstruction to uniqueness at first order, observed for the uncharged hyperbolic Schwarzschild--AdS black hole \cite{katona_uniqueness_2023}, persists in the charged case as well. That is, taking $Z^{(0)}=\xi$ and $\mathring\Psi \tilde\mu^{(1)}=2\mathring\nabla_{(a}\xi_{b)}$ for some harmonic one-form $\xi$ on the hyperbolic surface solves (\ref{eq_hyperbolic}) for $n=1$ and arbitrary $\Lambda r_0^2<-1$.
\end{enumerate}

\noindent{\bf Acknowledgements.} I would like to thank James Lucietti for the helpful conversations on this project, and his comments on the manuscript. This work is supported by an EPSRC studentship.

\section*{Statements and declarations}

\noindent {\bf Competing interests.}  The authors have no relevant financial or non-financial interests to disclose.
\\

\noindent{\bf Data availability.} Data sharing is not applicable to this article as no datasets were generated or analysed during the current study.
\appendix
\section{Einstein--Maxwell equations in Gaussian null coordinates}

In this section we list the relevant Einstein--Maxwell equations in Gaussian null coordinates, using the quantities introduced in (\ref{eq_GNmetric}-\ref{eq_extremalGN}). Here we assume that $\partial_v$ is Killing and preserves the Maxwell field, thus we omit any $v$-derivatives.

The $rv, ra, rr, ab$ components of the Einstein equations $\mathcal{E}:=\operatorname{Ric} -\Lambda g -8\pi T(\mathcal{F})=0$ respectively read 
\begin{align}
	\mathcal{E}_{rv}&:=\frac{1}{2\sqrt{\det \mu}}\partial_\rho\left[\sqrt{\det \mu}\left(\partial_\rho\phi-\beta^a\partial_\rho\beta_a\right)\right]+\frac{1}{2}\hat\nabla^a(\partial_\rho\beta_a)\nonumber\\
	&\qquad\qquad-\Lambda+2\Psi^2+2\Psi Z_a\beta^a-2\rho Z_aW^a+\frac{1}{2}\mathcal{F}^2=0\;, \label{eq_phi}\\
	\mathcal{E}_{ra}&:=\frac{1}{2\sqrt{\det \mu}}\partial_\rho\left[\sqrt{\det \mu}\partial_\rho\beta_a-\beta^b\partial_\rho\mu_{ab}\right] +\frac{1}{2}\hat\nabla^b\left(\partial_\rho\mu_{ab}\right)\nonumber\\
	&\qquad\qquad-\frac{1}{2}\hat\nabla_a\left(\tr\partial_\rho\mu\right)-2\Psi Z_a-2Z^bB_{ab}-2Z_aZ_b\beta^b=0\;, \label{eq_beta}\\
	\mathcal{E}_{rr}&:= -\frac{1}{2}\mu^{ab}\partial_\rho^2\mu_{ab} + \frac{1}{4}\mu^{ac}\mu^{bd}\partial_\rho\mu_{ab}\partial_\rho\mu_{cd}-2Z_aZ^a=0\;,\label{eq_trace}\\
	\mathcal{E}_{ab}&:= \frac{1}{2\sqrt{\det\mu}}\partial_\rho{\left[\sqrt{\det\mu}\left(2\hat\nabla_{(a}\beta_{b)}
	+\phi\partial_\rho\mu_{ab}-\beta^c\beta_c\partial_\rho\mu_{ab}\right)\right]}+\frac{1}{2}\hat\nabla_c\left(\beta^c\partial_\rho\mu_{ab}\right)
	\nonumber\\
	&\qquad\quad+\mathcal{R}_{ab}-\frac{1}{2}\left[\partial_\rho\beta_a-\beta^c\partial_\rho\mu_{ac}\right]
	\left[\partial_\rho\beta_b-\beta^d\partial_\rho\mu_{bd}\right] -(\partial_\rho\mu_{c(a})\hat\nabla^c\beta_{b)} \nonumber\\
	&\qquad\quad\;+ \frac{1}{2}(\beta^c\beta_c-\phi)(\partial_\rho\mu_{ac})(\partial_\rho\mu_{bd})\mu^{cd}-\Lambda\mu_{ab}-4\rho W_{(a}Z_{b)}-2Z_aZ_b(\beta^c\beta_c-\phi)\nonumber\\
	&\qquad\quad\;\;-4Z_{(a}B_{b)c}\beta^c-2B_{ac}B_b{}^c+\frac{1}{2}\mu_{ab}\mathcal{F}^2=0\; ,\label{eq_mu}
\end{align}
with
\begin{equation}
	\mathcal{F}^2 = -2 \Psi^2 -4\Psi Z_a\beta^a +4\rho W_aZ^a-4B_{ab}\beta^aZ^b+B_{ab}B^{ab} - 2(Z_a\beta^a)^2+2Z_aZ^a(\beta_b\beta^b-\phi)\;,
\end{equation}
where $\hat\nabla$ denotes the covariant derivative on $(S_{v,\rho}, \mu)$ and the indices $a,b,\dots$ are raised and lowered with $\mu$.

The Maxwell--equations $\nabla_\mu \mathcal{F}^{v\mu}=0$ and $\nabla_\mu \mathcal{F}^{a\mu}=0$ respectively yield
\begin{align}
	\mathcal{M}^v:=&-\frac{1}{\sqrt{\det\mu}}\partial_\rho\left[\sqrt{\det\mu}(\Psi+\beta_aZ^a)\right]+\hat\nabla_aZ^a=0\;, \label{eq_psi}\\
	\mathcal{M}^a:=&\frac{1}{\sqrt{\det \mu}}\partial_\rho\left[\sqrt{\det\mu}\left(\beta_bB^{ba}-\rho W^a-Z^a(\beta^2-\phi)+\beta^a(\Psi+\beta_bZ^b)\right)\right]\nonumber\\
	&\qquad\qquad\qquad\qquad+\hat\nabla_b\left(B^{ab}+Z^a\beta^b-Z^b\beta^a\right)=0\label{eq_Z}\;.
\end{align}

The Bianchi identities for $\mathcal{F}$ yield
\begin{align}
	\mathcal{B}_W:=&\hat\td \Psi- \partial_\rho(\rho W)=0\;,\label{eq_BianchiW}\\ 
	\mathcal{B}_B:=& \partial_\rho B - \hat \td Z=0\;,\label{eq_BianchiB}
\end{align}
where $\hat\td$ denotes the external derivative projected onto $S_{v,\rho}$, thus $W$ and $B$ are completely fixed by $\Psi$, $Z$ and the near-horizon data.

\bibliographystyle{sn-mathphys_modified}
\bibliography{ref}
%\printbibliography %Prints bibliography if biblatex used
\end{document}